\documentclass[letterpaper, 10 pt, conference]{ieeeconf}  

\IEEEoverridecommandlockouts                              

\usepackage[pdftex]{graphicx}

\usepackage{amsthm}
\usepackage{array}
\usepackage{url}
\usepackage{cite}
\usepackage{algorithmic}
\usepackage{graphicx}
\usepackage{textcomp}
\usepackage{xcolor}
\usepackage{amsmath}
\usepackage{amssymb}
\usepackage[english]{babel}

\usepackage{booktabs}
\usepackage{graphicx}
\usepackage{romannum}
\usepackage{epsfig}

\newtheorem{theorem}{Theorem}
\newtheorem{lemma}{Lemma}
\theoremstyle{remark}
\newtheorem{remark}{Remark}
\newtheorem{problem}{\bf Problem}
\theoremstyle{definition}
\newtheorem{definition}{Definition}
\title{\LARGE \bf
	Safe Control Synthesis Using Environmentally Robust Control Barrier Functions
}

\author{Vahid Hamdipoor$^{1}$, Nader Meskin$^{1}$, and Christos G. Cassandras$^{2}$
	\thanks{$^{1}$Vahid Hamdipoor and Nader Meskin are with the department of Electrical Engineering, Qatar University, Doha, Qatar {\tt\small nader.meskin@qu.edu.qa}.}%
	\thanks{$^{2}$Christos G. Cassandras is with the Division of Systems Engineering and Center for
		Information and Systems Engineering, Boston University, Brookline, MA,
		02446, USA {\tt\small cgc@bu.edu}.}%
			\thanks{This work was supported in part by NSF under grants ECCS-1931600, DMS-1664644, CNS-1645681, CNS-2149511, by AFOSR under grant FA9550-19-1-0158, by ARPA-E under grant DE-AR0001282, by the MathWorks, by the 
Red Hat-Boston University Collaboratory, and by NPRP grant (12S-0228-190177) from the Qatar National Research
Fund, a member of the Qatar Foundation (the statements made herein are
solely the responsibility of the authors).}
}

\begin{document}

	\maketitle
	\thispagestyle{empty}
	\pagestyle{empty}

	\begin{abstract}
		
		In this paper, we study a safe control design for dynamical systems in the presence of uncertainty in a dynamical environment. The worst-case error approach is considered to formulate robust Control Barrier Functions (CBFs) in an optimization-based control synthesis framework. It is first shown that environmentally robust CBF formulations result in second-order cone programs (SOCPs). Then, a novel scheme is presented to formulate robust CBFs which takes the nominally safe control as its desired control input in optimization-based control design and then tries to minimally modify it whenever the robust CBF constraint is violated. This proposed scheme leads to quadratic programs (QPs) which can be easily solved. Finally, the effectiveness of the proposed approach is demonstrated on an adaptive cruise control example.
		
	\end{abstract}

	\section{INTRODUCTION}
	
	Control Barrier Functions (CBFs) have emerged as a powerful means for guaranteeing control system safety in the form of set invariance \cite{ames2014control}. CBFs are often used with Control Lyapunov Functions (CLFs) to simultaneously ensure stability and safety of the system along with state and input constraints \cite{ames2016control}. This approach has been successfully implemented in numerous applications such as mobile robots \cite{gurriet2018online}, robotic manipulators \cite{cortez2019control}, robotic swarms \cite{wang2017safety}, aerial vehicles, racing drones \cite{singletary2022onboard}, and spacecraft docking \cite{breeden2021guaranteed}. CBFs have also been extensively used in the realm of the autonomous vehicles to generate a safe control input in problems such as cruise control, on-ramp merging \cite{ames2014control,xiao2021bridging}, signal free intersections \cite{xu2022general}, and the lane change \cite{chen2020cooperative}, to name a few.
	
	One of the challenges in designing a control input using CBFs is that controllers rely on models or measurements that are assumed to be perfect and free of uncertainty. However, models or measurements are usually uncertain or imperfect and this can result in an unsafe behaviour if not accounted for properly. This problem has been mainly addressed through the development of robust CBFs. However, there exist different perspectives in considering robustness in safety-critical control using CBFs. Generally, the proposed approaches in the literature consider either input disturbances \cite{kolathaya2018input,takano2020robust,dacs2022robust,wang2022disturbance,choi2021robust,buch2021robust,seiler2021control,alan2021safe,jankovic2018robust},   dynamic model mismatch \cite{long2022safe,nguyen2021robust}, or measurement errors \cite{dean2020guaranteeing,cosner2021measurement,yaghoubi2021risk,clark2021control,garg2021robust}. While the aforementioned works study robustness of CBF-based controllers, they do not take into account that the environment is dynamically changing and only consider a quasi-static environment. A good example of ``environment'' is the “other agents” in a multi-agent system with which the ego-agent is interacting, e.g., in the cruise control problem \cite{ames2014control} the lead vehicle can be considered as the environment. 
	
	There exist limited works that consider the effect of a dynamic environment in the design of a safe control input via CBFs  \cite{wu2016safety,he2021rule,chalaki2022barrier,long2021learning,molnar2022safety}. In \cite{wu2016safety,he2021rule}, the effect of a dynamic environment is considered through the notion of time-varying CBFs in which the time-derivative of the CBF is also added to the CBF constraint. Even though the authors in \cite{wu2016safety,he2021rule,chalaki2022barrier} have taken a dynamic environment into consideration, they do not study robustness and assume that perfect information on the environment is available. In \cite{long2021learning}, safe navigation in unknown environments is studied, where on-board range sensing is utilized to construct CBFs online. In the recent work \cite{molnar2022safety}, the notion of {\em Environmental Control Barrier Functions (ECBFs)} has been introduced and robust ECBFs against errors in the environment, in particular a time-delay, are investigated. In this paper, inspired by \cite{molnar2022safety}, we study the notion of \emph{Environmentally Robust Control Barrier Function}  ({\em ER-CBFs}). Similar to \cite{molnar2022safety}, worst-case error-based ER-CBFs are considered, however, unlike \cite{molnar2022safety} ($i$) the Lipschitz constant is not used to define the worst-case error in an ER-CBF constraint and ($ii$) it is not assumed that the dynamics of the environment are known. Moreover, the original quadratic program (QP) for the control synthesis is converted to a second-order cone program (SOCP). Finally, we propose our main robust control scheme, by exploiting the closed-form solution of the original QP.
	
	In summary, the contributions of this paper are as follows. An ER-CBF is introduced  based on errors in the nominal CBF resulting from errors in a dynamically changing environment. This ER-CBF contains the norm control input (which is the optimization decision variable) and leads to the formulation of a control synthesis problem as a SOCP. Secondly,  the nominally safe input obtained from a nominal CBF is considered as the desired input for the SOCP and a new ER-CBF is introduced which does not depend on the norm of the control input, hence the control synthesis problem leads to solving a QP. Then,  the explicit solution to this QP is obtained. Finally,  the effectiveness of the presented results is demonstrated in an adaptive cruise control example.
	
	The rest of the paper is organized as follows. Preliminaries on CBFs and CLFs are reviewed in Section \Romannum{2}. The problem formulation is discussed in Section \Romannum{3}. Main results are presented in Section \Romannum{4}. An adaptive cruise control example is visited in Section \Romannum{5}.  Numerical simulations on the cruise control example are presented in Section \Romannum{6}. Finally, concluding remarks are provided in the last section.

\section{Preliminaries}
Consider a nonlinear control affine system 
\begin{equation}\label{eq:system}
	\dot{x}= f(x) + g(x)u,
\end{equation}
where $x \in \mathcal{X} \subset \mathbb{R}^n$,  $u \in  \mathbb{R}^m$ are state and control input, respectively, and $f:\mathbb{R}^n \rightarrow \mathbb{R}^n $ and $g:\mathbb{R}^n \rightarrow \mathbb{R}^{n\times m} $ are locally Lipschitz continuous functions. To establish exponential stability of system (\ref{eq:system}) without having to define an explicit feedback controller, the notion of a {\em Control Lyapunov Function (CLF)} is introduced.
\begin{definition}(Control Lyapunov Function (CLF) \cite{ames2019control})
	A continuously differentiable function $V:\mathbb{R}^n \rightarrow \mathbb{R} $ is called a Control Lyapunov Function (CLF) for system (\ref{eq:system}) if there exist positive constants $  {c}_1$, $ {c}_2$, and $  {c}_3$ such that for $\forall x \in \mathcal{X}$,
	\begin{equation}
		{c}_1 \| x \|^2 \leq V(x) \leq  {c}_2 \| x \|^2,
	\end{equation}
	\begin{equation}\label{eq:CLF}
		\inf_{u}[L_fV(x) + L_gV(x)u +   {c}_3 V(x)] \leq 0,
	\end{equation}
	where $L_f V(x) \triangleq \nabla_x V(x)f(x) $ and $L_g V(x) \triangleq \nabla_x V(x)g(x)$ are Lie derivatives of $V(x)$ along the vector fields $f$ and $g$, respectively.
	\label{def:CLF}
\end{definition}
%
Given a stabilizing CLF $V(x)$ as in Definition \ref{def:CLF}, any Lipschitz continuous controller $u(t) \in \mathbb{R}^m$ that satisfies (\ref{eq:CLF}), 
$\forall t \geq 0$ exponentially stabilizes system (\ref{eq:system}) to the origin.

For the safe operation of  system (\ref{eq:system}), a safe set $\mathcal{C} \subset \mathcal{X} \subset \mathbb{R}^n$ is defined as a superlevel set of a differentiable function $h:\mathbb{R}^n  \times [0,+\infty)  \rightarrow \mathbb{R}$, i.e.
\begin{equation}\label{eq:safe_set}
	\mathcal{C} = \{ x \in \mathbb{R}^n: h(x,t) \geq 0 \}.
\end{equation}
To ensure that system (\ref{eq:system}) remains in the safety set $\mathcal{C}$, the notion of {\em Control Barrier Function (CBF)} can be defined as follows, in a similar manner as the concept of CLF:
\begin{definition}(Control Barrier Function (CBF)  \cite{he2021rule})
	Given the set $\mathcal{C}$ defined in (\ref{eq:safe_set}), a continuously differentiable function $h:\mathbb{R}^n \times [0,+\infty) \rightarrow \mathbb{R} $ is called as a control barrier function (CBF) for system (\ref{eq:system}), if there exists an extended class $\mathcal{K}_{\infty}$ function\footnote{An extended class $\mathcal{K}_{\infty}$ function is a function $\alpha: \mathbb{R} \rightarrow \mathbb{R}$ that is strictly increasing and $\alpha(0) = 0$. } $\alpha$ for all $x \in \mathcal{X}$ such that 
	\begin{equation}\label{eq:CBF}
		\sup_{u}[\frac{\partial h(x,t)}{\partial t} + L_fh(x,t) + L_gh(x,t)u + \alpha (h(x,t))] \geq 0
	\end{equation}
	where $L_f h(x,t) \triangleq \nabla_x h(x,t) f(x) $ and $L_g h(x,t) \triangleq \nabla_x h(x,t)g(x)$ are Lie derivatives of $h(x,t)$ along the vector fields $f$ and $g$, respectively.
	\label{def:CBF}
\end{definition}
Notice that CBFs used in this paper are time-varying functions. When $h(x,t)$ is time-invariant, it is denoted as $h(x)$ and the condition (\ref{eq:CBF}) is simplified \cite{ames2019control} as follows:
\begin{equation}
	\sup_{u}[ L_fh(x) + L_gh(x)u +   \alpha (h(x))] \geq 0.
\end{equation}
In this study, a particular choice of extended class $\mathcal{K}_{\infty}$ function having the form of $\alpha(h(x,t)) = \nu h(x,t)$ is considered where $\nu \geq 0$ is a CBF design parameter which controls the system behaviours near the boundary of $h(x,t) = 0$. Safety of a set for a specific dynamical system is often stated in terms of the {\em forward invariance} property of that set with respect to the dynamical system:
\begin{definition}
	A set $\mathcal{C} \subset \mathbb{R}^n$ is forward invariant for system (\ref{eq:system}) if its solution starting at $x(0) \in \mathcal{C}$ satisfies $x(t) \in \mathcal{C}$ , $\forall t \geq 0$.
\end{definition}
\begin{definition}
	System (\ref{eq:system}) is safe with respect to set $\mathcal{C}$ defined in (\ref{eq:safe_set}), if the set $\mathcal{C}$ is forward invariant.
\end{definition}
\begin{theorem}[\cite{ames2019control}]\label{thm:CBF}
	Given a CBF $h(x,t)$ as in Definition (\ref{def:CBF}) with the associated set $\mathcal{C}$, any Lipschitz continuous controller $u(t)$ that satisfies (\ref{eq:CBF}), 
	$\forall t \geq 0$ renders $\mathcal{C}$ forward invariant for system (\ref{eq:system}).
\end{theorem}
The advantage of CLF and CBF formulations is that they allow the unification of control objectives (represented by CLFs) that are regulated to yield trajectories within desired sets (as enforced by CBFs).
Given a CLF $V(x)$ and a CBF $h(x,t)$, they can be combined into a single controller by using a sequence of Quadratic Programs (QPs).
\section{Problem Statement}
Most of the works to date on designing robust CBFs that account for uncertainty consider a quasi-static environment for which it is assumed that a perfect knowledge of the environment is available. Unlike previous studies, in this paper, similar to \cite{molnar2022safety}, we consider a robust safety-critical control problem with regard to existing uncertainty in the state of the dynamic environment. If $x_s \in \mathbb{R}^p$ represents the state of the dynamic environment, we redefine the safety set $ \mathcal{C}$ and the time-varying CBF $h(x,x_s)$ as 
\begin{equation}\label{eq:safe_set_2}
	\mathcal{C}_s = \{ x\in \mathbb{R}^n, x_s\in \mathbb{R}^p: h(x,x_s) \geq 0 \},
\end{equation}
and $h:\mathbb{R}^n \times \mathbb{R}^p \rightarrow \mathbb{R}$. Consider $x_s = \hat{x}_s+e_s$ and $\dot{x}_s = \dot{\hat{x}}_s+\dot{e}_s$ where $\hat{x}_s$ denotes the measured state of the environment and $e_s$ corresponds to the uncertainty in the measured state. We assume that $e_s$ and $\dot{e}_s$ are bounded, i.e., $||e_s|| <\mathcal{E}_s$ and $||\dot{e}_s|| <\dot{\mathcal{E}}_s$. 
Our goal is to design a feedback controller $u(x,\hat{x}_s)$ for a given system (\ref{eq:system}), such that the trajectories of the closed-loop system remain inside the safety set defined in (\ref{eq:safe_set_2}).
Note that since the dependence of $h$ on time comes only through $x_s(t)$, we wil henceforth write $h(x,x_s)$ instead of $h(x,t)$ and drop $t$ from $x_s(t)$ in order to ease notation.
Next, an Environmentally Robust CBF (ER-CBF) is defined.
\begin{definition}\label{ER-CBF_def}
	A function $h(x,x_s)$ is an Environmentally Robust Control Barrier Function (ER-CBF) for system (\ref{eq:system}) if there exists an extended class $\mathcal{K}$ function $\alpha$ such that for all $x \in \mathcal{X}$  and $\hat{x}_s \in \mathcal{X}_s$  
	\begin{align}\label{def:ER-CBF}
		\sup_{u}&[\frac{\partial h(x,\hat{x}_s)}{\partial t} + L_fh(x,\hat{x}_s) + L_g h(x,\hat{x}_s)u\\ 
		&+ \alpha (h(x,\hat{x}_s)) +\Delta(x,x_s,\hat{x}_s,u) ] \geq 0, \notag
	\end{align}
	where $\Delta(x,x_s,\hat{x}_s,u)\in \mathbb{R}$ is the residual term which appears due to the difference between $x_s$ and $\hat{x}_s$, and $L_f h(x,\hat{x}_s) \triangleq \nabla_x h(x,\hat{x}_s) f(x) $ and $L_g h(x,\hat{x}_s) \triangleq \nabla_x h(x,\hat{x}_s)g(x)$.
\end{definition}
\begin{problem}[Safety under worst-case uncertainty in environment]
	Given a surrounding dynamical system estimate $\hat{x}_s$ with the error bounds of $||e_s|| <\mathcal{E}_s$ and $||\dot{e}_s|| <\dot{\mathcal{E}}_s$, design a feedback controller $u$ for system (\ref{eq:system}) such that the set $\mathcal{C}_s$ is rendered forward invariant for system (\ref{eq:system}) for all $x \in \mathcal{X}$ and  $x_s \in \mathcal{X}_s$.
\end{problem}
\begin{problem}[Robust safety of nominally safe controller]
	Given a surrounding dynamical system estimate $\hat{x}_s$ with the error bounds of $||e_s|| <\mathcal{E}_s$ and $||\dot{e}_s|| <\dot{\mathcal{E}}_s$ and a nominal safe control $u$ for system (\ref{eq:system}) which satisfies (\ref{eq:CBF}), design a robust feedback control $u_\text{rob}$ by minimally modifying $u$, such that the set $\mathcal{C}_s$ is rendered forward invariant for system (\ref{eq:system}) for all $x \in \mathcal{X}$ and  $x_s \in \mathcal{X}_s$.
\end{problem}

Problem 1 seeks to find a robust safe control input based on a given desired control and error in the state of the environment, while Problem 2 only tries to robustify a given nominal safe control with respect to error in the surrounding environment.  In the next section, we present our main results regarding the problems described in this section.

\section{Results}
%
First, the CBF error is quantified due to estimating the state of the surrounding dynamical system, its gradient, and its time-derivative.
Let 
\begin{align*}
	e_h(x,x_s,\hat{x}_s) &\triangleq h(x,x_s) - h(x,\hat{x}_s),\\
	e_{\nabla h}(x,x_s,\hat{x}_s) &\triangleq \nabla h(x,x_s) -\nabla h(x,\hat{x}_s),\\
	e_{\frac{\partial h}{\partial t}}(x,x_s,\hat{x}_s) &\triangleq \frac{\partial h(x,x_s)}{\partial t} - \frac{\partial h(x,\hat{x}_s)}{\partial t}.
\end{align*}
Since $e_s$ and $\dot{e}_s$ are bounded, it follows that $e_h$, $e_{\nabla h}$, and $ e_{\frac{\partial h}{\partial t}}$ are bounded as well. In this study, it is intended to design a control input for (\ref{eq:system}) based on the  worst-case bounds of $e_h$, $e_{\nabla h}$, and $e_{\frac{\partial h}{\partial t}}$ and the worst-case uncertainties for the above errors are considered as follows:
\begin{align}\label{eq:e_h}
	e_h^*(x,x_s,\hat{x}_s) &= \min_{||e_s|| <\mathcal{E}_s, ||\dot{e}_s|| <\dot{\mathcal{E}}_s} e_h(x,x_s,\hat{x}_s),\nonumber \\
	e_{\nabla h}^*(x,x_s,\hat{x}_s) &= \max_{||e_s|| <\mathcal{E}_s,||\dot{e}_s|| <\dot{\mathcal{E}}_s} \|e_{\nabla h}(x,x_s,\hat{x}_s)\|,\\
	e_{\frac{\partial h}{\partial t}}^*(x,x_s,\hat{x}_s) &= \min_{||e_s|| <\mathcal{E}_s,||\dot{e}_s|| <\dot{\mathcal{E}}_s}e_{\frac{\partial h}{\partial t}}(x,x_s,\hat{x}_s)\nonumber.
\end{align}
where it should be noted that $e_h^*$ and $ e_{\frac{\partial h}{\partial t}}^*$ is obtained by minimizing, while $e_{\nabla h}^*$ is obtained by maximization. The reason for this is that (as shown in the next result) $e_h^*$ and $ e_{\frac{\partial h}{\partial t}}^*$ appear as additive error terms in the residual term of the ER-CBF constraint, while $e_{\nabla h}^*$ appears as a multiplicative error term.
From now on, in order to ease notation,  the arguments $(x,x_s,\hat{x}_s)$ are omitted from $e_h^*$, $e_{\nabla h}^*$ and $ e_{\frac{\partial h}{\partial t}}^*$.
The next theorem provides a solution for \emph{Problem 1} and presents a sufficient condition in which an ER-CBF preserves system (\ref{eq:system}) forward invariance in the presence of bounded uncertainties in the state of the surrounding dynamical environment.
\begin{theorem}\label{thm:er_CBF} 
	If $h(x,x_s)$ is an ER-CBF for system (\ref{eq:system}) with the residual term
	\begin{equation}\label{eq:residu1}
		\Delta(x,x_s,\hat{x}_s,u)=  e_{\frac{\partial h}{\partial t}}^* +  \alpha( e_h^*)-e_{\nabla h}^*||f(x)+g(x)u||,
	\end{equation}
	where  $e_h^*$, $e_{\nabla h}^*$, and $e_{\frac{\partial h}{\partial t}}^*$ are defined in (\ref{eq:e_h}), 
	then any Lipschitz continuous controller $u$ that satisfies \eqref{def:ER-CBF} will render system (\ref{eq:system}) forward invariant with respect to $\mathcal{C}_s$.
\end{theorem}
\begin{proof}
According to Theorem \ref{thm:CBF}, if there exists a control input $u\in \mathbb{R}^m$ such that 
	\begin{equation}\label{eq:ecbf_constraint}
		\frac{\partial h(x,x_s)}{\partial t}+\nabla_x h(x,x_s)(f(x)+g(x)u)+ \alpha (h(x,x_s)) \geq 0,
	\end{equation}
then, the set $\mathcal{C}_s$ is forward invaraint.  By rewriting  $h(x,x_s) = h(x,\hat{x}_s)+e_h(x,x_s,\hat{x}_s)$, $  \nabla_x h(x,x_s) =\nabla_x h(x,\hat{x}_s)+e_{\nabla h}(x,x_s,\hat{x}_s) $, and $ \frac{\partial h(x,x_s)}{\partial t} = \frac{\partial h(x,\hat{x}_s)}{\partial t}+e_{\frac{\partial h}{\partial t}}(x,x_s,\hat{x}_s)$  in (\ref{eq:ecbf_constraint}), it is needed that a lower bound of the left-hand side expression which is still positive in the presence of the worst-case error in the state of the surrounding dynamical system, i.e.,
	\begin{align}\label{eq:ecbf_constraint2}
		&\min_{||e_s|| <\mathcal{E}_s, ||\dot{e}_s|| <\dot{\mathcal{E}}_s}[\frac{\partial h(x,\hat{x}_s)}{\partial t}+e_{\frac{\partial h}{\partial t}}(x,x_s,\hat{x}_s)+\\
		&\nabla_x h(x,\hat{x}_s)(f(x)+g(x)u)+e_{\nabla h}(x,x_s,\hat{x}_s)(f(x)+g(x)u) \notag\\
		&+\alpha ( h(x,\hat{x}_s)+e_h(x,x_s,\hat{x}_s))] \geq 0 \notag.
	\end{align}
Recalling that we consider $\alpha(h(x,t)) = \nu h(x,t)$, the extended class $\mathcal{K}$ function $\alpha$ can be selected such that $\alpha ( h(x,\hat{x}_s)+e_h(x,x_s,\hat{x}_s)) = \alpha ( h(x,\hat{x}_s)) +\alpha(e_h(x,x_s,\hat{x}_s))$. Moreover, due to the Cauchy-Shwarz inequality, $-\|e_{\nabla h}(x,x_s,\hat{x}_s)\| \|f(x)+g(x)u\| \leq e_{\nabla h}(x,x_s,\hat{x}_s)(f(x)+g(x)u)$. By minimizing each term individually in (\ref{eq:ecbf_constraint2}), it follows that
	\begin{align*}
		&\frac{\partial h(x,\hat{x}_s)}{\partial t}+\nabla_x h(x,\hat{x}_s)(f(x)+g(x)u)+\alpha(h(x,\hat{x}_s))\\
		&\underbrace{-e_{\nabla h}^*\|f(x)+g(x)u\|+e_{\frac{\partial h}{\partial t}}^*+ \alpha(e_h^*)}_{\Delta(x,x_s,\hat{x}_s,u)} \geq 0,
	\end{align*}
	which, based on Definition \ref{ER-CBF_def}, implies that $h(x,x_s)$ is an ER-CBF for system (\ref{eq:system}). 
\end{proof}
To obtain the safe input $u$ with respect to a desired control input $u_{\text{des}}$, the following optimization problem is required 
to be solved at each time step to determine a value $u$ kept constant over this time step:
\begin{equation*}
u^* = \textrm{arg} \min_{u} \quad \frac{1}{2} \|u-u_{\text{des}}\|^2 
\end{equation*}
\begin{align}\label{eq:QP2}
\textrm{s.t.}  &\qquad  \frac{\partial h(x,\hat{x}_s)}{\partial t}+L_fh(x,\hat{x}_s)+L_gh(x,\hat{x}_s)u \,+\\
&\alpha(h(x,\hat{x}_s))- e_{\nabla h}^*\|f(x)+g(x)u\|+e_{\frac{\partial h}{\partial t}}^*+ \alpha(e_h^*) \geq 0 \notag
\end{align}
Since the optimization variable $u$ appears in $\|f(x)+g(x)u\|$ above, this problem is no longer a QP. By writing $\frac{1}{2} \|u-u_{\text{des}}\|^2 =\frac{1}{2}(u-u_{\text{des}})^T(u-u_{\text{des}}) =  \frac{1}{2} \|u\|^2 +u_{\text{des}}^Tu + \frac{1}{2}\|u_{\text{des}}\|^2$ and removing the term $\frac{1}{2}\|u_{\text{des}}\|^2$ from the objective function as it is a constant, a slack variable $q$ can be used to restate the  optimization problem  as follows:
\begin{align*}
[u^*,q^*]^T &= \textrm{arg} \min_{u,q } \quad  q - u_{\text{des}}^Tu\\
&\textrm{s.t.}  \quad (\ref{eq:QP2}), \,
\frac{1}{2} \|u\|^2 \leq q,
\end{align*}
with the decision variables $u$ and $q\geq 0$. This problem is a second-order cone program (SOCP). The second constraint above can be written as a rotated second-order cone condition \cite{calafiore2014optimization}, which  leads to the following SOCP:
\begin{align*}
[u^*,q^*]^T &= \textrm{arg} \min_{u,q} \quad  q - u_{\text{des}}^Tu\\
&\textrm{s.t.}  \quad  (\ref{eq:QP2}), \,
\left\|\left[\begin{matrix}
	\sqrt{2}u\\
	q-1
\end{matrix}\right]\right\| \leq q+1.
\end{align*}
The difficulty with SOCPs is that, unlike QPs, SOCPs tend to be infeasible very easily \cite{castaneda2021pointwise} and, therefore, sometimes another slack variable is added to the constraint (\ref{eq:QP2}) to ensure the feasibility of the SOCP \cite{dean2020guaranteeing}, which can deteriorate the safe input. Another problem with SOCPs, as reported in \cite{long2022safe}, is that the computation time to solve them is higher than that of QPs, which can affect the real-time control synthesis as well. In \cite{buch2021robust}, it is assumed that the control input $u$ is not overly restrictive, thus permitting sufficient control authority to preserve safety in the presence of uncertainty.

To circumvent solving a SOCP for ER-CBFs at each time step, we propose a novel ER-CBF which is based on computing the control input using the nominal CBF and modifying it whenever the ER-CBF constraint in (\ref{eq:QP2}) is violated. In this way, the control input produced by the nominal CBF is considered as the desired input and it is modified minimally so as to robustly guarantee safety. First, we denote the control input $u$ in the optimization problem with nominal CBFs and ER-CBFs as $u_{\text{nom}}$ and $u_{\text{rob}}$, respectively. To begin with, we present the closed-form solution for the nominal safe input framed as the following QP (solved at each time step):
\begin{align*}
 u_{\text{nom}}^* &= \textrm{arg} \min_{u_{\text{nom}}} \quad  \|u_{\text{nom}} -u_{\text{des}} \|^2 \tag{CBF-QP} \label{CBF-QP}\\ 
\textrm{s.t.}  \quad  &\frac{\partial h(x,x_s)}{\partial t}+  L_fh(x,x_s) \\
&+ L_gh(x,x_s)u_{\text{nom}} + \alpha (h(x,x_s)) \geq 0
\end{align*}
For this QP, the closed-form solution can be obtained using the following Lemma.
\begin{lemma}(\cite{xu2015robustness}, \cite{singletary2021safety})\label{lem:QP-Closed-form}
Consider $h(x,x_s)$ as a CBF for system (\ref{eq:system}) with $L_gh(x,x_s) \neq 0$. The explicit solution to (\ref{CBF-QP}) is given by $u_{\text{nom}}^\ast = u_{\text{des}} + u_{s}$ where 
\begin{equation}
	u_s = 
	\begin{cases}
		-\frac{L_gh(x,x_s)^T}{\|L_gh(x,x_s)\|^2} \Phi_{\text{nom}}(x,x_s,u_{\text{des}}) \quad \textrm{if}\\  \qquad \qquad \qquad \qquad \Phi_{\text{nom}}(x,x_s,u_{\text{des}}) < 0,\\
		0 \qquad \textrm{if}\qquad \qquad \quad  \Phi_{\text{nom}}(x,x_s,u_{\text{des}}) \geq 0
	\end{cases}\notag
\end{equation}
with $ \Phi_{\text{nom}}(x,x_s,u) = \frac{\partial h(x,x_s)}{\partial t}+  L_fh(x,x_s)+ L_gh(x,x_s)u + \alpha(h(x,x_s))$.
\end{lemma}
Our method to synthesize the control input using ER-CBFs is inspired by Lemma \ref{lem:QP-Closed-form}, where we intend to use $u_{\text{nom}}^\ast$ as the desired input and find an extra control effort such that the resulting input ensures robustness in the presence of environmental errors. 
First, let
\begin{align*}
&\Phi_\text{rob}(x,\hat{x}_s,u)= \frac{\partial h(x,\hat{x}_s)}{\partial t}+L_fh(x,\hat{x}_s)+L_gh(x,\hat{x}_s)u \,+\\
&\alpha(h(x,\hat{x}_s))- e_{\nabla h}^* \|f(x)+g(x)u\|+e_{\frac{\partial h}{\partial t}}^*+ \alpha(e_h^*).
\end{align*}
Then, in order to design an environmentally safe controller with ER-CBFs, one needs to consider the following SOCP:
\begin{align}
&u_{\text{rob}}^* = \textrm{arg} \min_{u_{\text{rob}}} \quad  \|u_{\text{rob}} -u_{\text{des}} \|^2 \tag{ER-CBF-SOCP} \label{ER-CBF-SOCP}\\ \nonumber
&\textrm{s.t.} \quad \Phi_\text{rob}(x,\hat{x}_s,u_{\textrm{rob}}) \, \geq 0 \notag
\end{align}

Since $u^*_{\text{nom}}$ which is the solution of (\ref{CBF-QP}), will be taken as the desired input for (\ref{ER-CBF-SOCP}), from now on, the desired input for (\ref{ER-CBF-SOCP}) is denoted by $u^{(s)}_{\text{des}}$ to distinguish it with $u_{\text{des}}$ in  (\ref{CBF-QP}). 
Let us define 
$u_\delta \triangleq u_{\text{rob}}^* - u_{\text{nom}}^*$
as the amount of change in the control input $u_{\text{nom}}^*$ such that robustness to the environmental uncertainties is guaranteed. Next, based on  Lemma \ref{lem:QP-Closed-form}, an upper bound on $u_\delta$ is obtained. Note that $u_s$ in Lemma 1 for the case where $u \in \mathbb{R}$ reduces to 
\begin{equation}
	u_s = 
	\begin{cases}
		-\frac{\Phi_{\text{nom}}(x,x_s,u_{\text{des}})}{L_gh(x,x_s)}  \: &\textrm{if} \quad \Phi_{\text{nom}}(x,x_s,u_{\text{des}}) < 0,\\
		0 \qquad &\textrm{if}  \quad  \Phi_{\text{nom}}(x,x_s,u_{\text{des}}) \geq 0.
	\end{cases}\notag
\end{equation}
%
\begin{theorem}\label{thm:u_delta}
Consider $u^*_{\text{nom}}$ to be the solution to (\ref{CBF-QP}) for system (\ref{eq:system}) with $u \in \mathbb{R}$ and let it be the desired control input for (\ref{ER-CBF-SOCP}), i.e., $u^{(s)}_{\text{des}} = u^*_{\text{nom}}$. Define $u_\delta \triangleq u_{\text{rob}}^* - u_{\text{nom}}^*$, with $u_{\text{rob}}^*$ being the solution for (\ref{ER-CBF-SOCP}), and let $e_{\frac{\partial h}{\partial t}}^*$, $ e_h^*$, and $e_{\nabla h}^*$ be obtained from \eqref{eq:e_h} at $\|e_s\| = \mathcal{E}_s$ and $\|\dot{e}_s\| =\dot{\mathcal{E}}_s $. Assuming that $L_gh(x,\hat{x}_s) \neq e^*_{\nabla}\|g(x)\|$, then
{\scriptsize
\begin{align*}
		|u_\delta| \leq \max \left\{\left|\frac{-\Phi_\text{rob}(x,\hat{x}_s,u^*_{\text{nom}})}{L_gh(x,\hat{x}_s)+e^*_{\nabla h}\|g(x)\|}\right|, \left|\frac{-\Phi_\text{rob}(x,\hat{x}_s,u^*_{\text{nom}})}{L_gh(x,\hat{x}_s)-e^*_{\nabla h}\|g(x)\|}\right| \right\}.
\end{align*}
}
\end{theorem}
\begin{proof}
Since  $e_{\frac{\partial h}{\partial t}}^*$, $ e_h^*$, and $e_{\nabla h}^*$ are obtained at $\|e_s\| = \mathcal{E}_s$ and $\|\dot{e}_s\| =\dot{\mathcal{E}}_s $, there exits $\delta h \in \mathbb{R}$ such that for ER-CBF $h(x,x_s)$ in (\ref{ER-CBF-SOCP}), it follows:
\begin{align}
&	\frac{\partial (h(x,\hat{x}_s)+ \delta h)}{\partial t}+L_f\left(h(x,\hat{x}_s\right)+ \delta h)+L_g(h(x,\hat{x}_s) \notag\\
&	+ \delta h)u+\alpha(h(x,\hat{x}_s)+ \delta h) = \Phi_\text{rob}(x,\hat{x}_s,u).\label{eq:delta_h}
\end{align}
Then, by considering the left-hand side of (\ref{eq:delta_h}) and using Lemma \ref{lem:QP-Closed-form} with $u^{(s)}_{\text{des}} = u_{\text{nom}}^*$ and noticing the fact that $u \in \mathbb{R}$, it follows that:
\begin{equation}
	\begin{cases}
		u_{\text{rob}}^* =  u_{\text{nom}}^* \: &\textrm{if} \:  \Phi_\text{rob}(x,\hat{x}_s,u^*_{\text{nom}}) \geq0\\
		u_{\text{rob}}^* =  u_{\text{nom}}^*-\frac{\Phi_\text{rob}(x,\hat{x}_s,u^*_{\text{nom}})}{L_g(h(x,\hat{x}_s)+\delta h)} \: &\textrm{if} \:  \Phi_\text{rob}(x,\hat{x}_s,u^*_{\text{nom}}) < 0
	\end{cases}
\end{equation}
Thus, for the case $\Phi_\text{rob}(x,\hat{x}_s,u^*_{\text{nom}}) \geq0$, $u_\delta$ is zero,  and for the case where $\Phi_\text{rob}(x,\hat{x}_s,u^*_{\text{nom}}) <0$, it follows: 
\begin{equation}\label{eq:u_delta}
u_\delta =- \frac{\Phi_\text{rob}(x,\hat{x}_s,u^*_{\text{nom}})}{L_g(h(x,\hat{x}_s)+\delta h)} =- \frac{\Phi_\text{rob}(x,\hat{x}_s,u^*_{\text{nom}})}{L_gh(x,\hat{x}_s)+L_g\delta h}.
\end{equation}
Note that $ L_g\delta h = \nabla_x \delta h g(x)$. Although, the exact value for $ L_g\delta h$ is not known, recall that $\|\nabla_x \delta h\| = e_\nabla ^*$ in Theorem \ref{thm:er_CBF}. Since $u \in \mathbb{R}$, it also follows that $L_gh(x,\hat{x}_s) \in \mathbb{R}$ and $ L_g\delta h \in \mathbb{R}$. Therefore, 
\begin{equation}
	| L_g\delta h | = |\nabla_x \delta h g(x)| =e_\nabla ^* \|g(x)\|.
\end{equation}
By replacing possible values of $L_g\delta h$  in \eqref{eq:u_delta}, i.e., $L_g\delta h  =\pm e_\nabla ^* \|g(x)\|$,  it follows that:
{\scriptsize
\begin{align}
	|u_\delta| \leq \max \left\{\left|\frac{-\Phi_\text{rob}(x,\hat{x}_s,u^*_{\text{nom}})}{L_gh(x,\hat{x}_s)+e^*_{\nabla h}\|g(x)\|}\right|,\notag  \left|\frac{-\Phi_\text{rob}(x,\hat{x}_s,u^*_{\text{nom}})}{L_gh(x,\hat{x}_s)-e^*_{\nabla h}\|g(x)\|}\right| \right\}.\notag
\end{align}
}
\end{proof}
\begin{remark}
	If the lower bound considered for \eqref{eq:ecbf_constraint2} in the proof of Theorem \ref{thm:er_CBF} is tight, we can always find such $\delta h$ in \eqref{eq:delta_h}. In other words, if $e_{\frac{\partial h}{\partial t}}^*$, $ e_h^*$, and $e_{\nabla h}^*$ which are obtained by minimizing each term individually in \eqref{eq:ecbf_constraint2}, have a common optimizer, such $\delta h$ can always be found. That is the reason it is assumed in Theorem \ref{thm:u_delta} that $e_{\frac{\partial h}{\partial t}}^*$, $ e_h^*$, and $e_{\nabla h}^*$ are obtained at $\|e_s\| = \mathcal{E}_s$ and $\|\dot{e}_s\| =\dot{\mathcal{E}}_s $. Note that this is not a strong assumption and $e_{\frac{\partial h}{\partial t}}^*$, $ e_h^*$, and $e_{\nabla h}^*$ most of the time are obtained in extreme situations.
\end{remark}
Theorem \ref{thm:u_delta} finds an upper bound on the size of the necessary modification in the nominal safe input to meet the robustness condition. 
Using the results of Theorem \ref{thm:u_delta},  a new constraint is established for the ER-CBF which is no longer dependent on the norm of the control input, hence  the control design problem can be formulated as a QP rather than a SOCP. 
\begin{theorem}\label{thm:er_CBF2}
Let $e_{\frac{\partial h}{\partial t}}^*$, $ e_h^*$, and $e_{\nabla h}^*$ be defined as in (\ref{eq:e_h}).
If $h(x,x_s)$ is an ER-CBF for system (\ref{eq:system}) with  $u \in \mathbb{R}$ and the residual term $	\Delta(x,x_s,\hat{x}_s,u)=  e_{\frac{\partial h}{\partial t}}^* +  \alpha( e_h^*)-e_{\nabla h}^*(||f(x)+g(x)u_{\text{nom}}|| +\|g(x)\| \Bar{u}_\delta)$,
where $u_{\text{nom}}^*$ is the nominal safe input obtained by solving (\ref{CBF-QP}) and
{\footnotesize
\begin{align*}
	\Bar{u}_\delta \triangleq  \max \left\{\left|\frac{-\Phi_\text{rob}(x,\hat{x}_s,u^*_{\text{nom}})}{L_gh(x,\hat{x}_s)+e^*_{\nabla h}\|g(x)\|}\right|, \notag 
	 \left|\frac{-\Phi_\text{rob}(x,\hat{x}_s,u^*_{\text{nom}})}{L_gh(x,\hat{x}_s)-e^*_{\nabla h}\|g(x)\|}\right|\right\},\notag
\end{align*} }
then any Lipschitz continuous controller $u $ that satisfies \eqref{def:ER-CBF} will render system (\ref{eq:system}) forward invariant with respect to $\mathcal{C}_s$.
\end{theorem}
\begin{proof}
According to Theorem \ref{thm:er_CBF}, if $\Delta(x,x_s,\hat{x}_s,u)=  e_{\frac{\partial h}{\partial t}}^* +  \alpha( e_h^*)-e_{\nabla h}^*(||f(x)+g(x)u|| )$, then any Lipschitz continuous controller $u\in \mathbb{R}$  that satisfies \eqref{def:ER-CBF} will render system (\ref{eq:system}) forward invariant with respect to $\mathcal{C}_s$. Next, we want to get rid of $u$ in this residual term by exploiting the result of Theorem \ref{thm:u_delta}. Considering $u^*_{\text{rob}} = u_{\text{nom}}^\ast + u_\delta$, it follows that:
\begin{align*}
	&e_{\frac{\partial h}{\partial t}}^* +  \alpha( e_h^*)-e_{\nabla h}^*(||f(x)+g(x)( u_{\text{nom}}^\ast + u_\delta)|| ) \geq \\
	& e_{\frac{\partial h}{\partial t}}^* +  \alpha( e_h^*)-e_{\nabla h}^*(||f(x)+g(x) u_{\text{nom}}^\ast\| +  |u_\delta|\|g(x)\|).
\end{align*}
Due to Theorem \ref{thm:u_delta}, $ |u_\delta| \leq \Bar{u}_\delta = \max \{|-\frac{\Phi_\text{rob}(x,\hat{x}_s,u^*_{\text{nom}})}{L_gh(x,\hat{x}_s)+e^*_{\nabla h}\|g\|}|,|-\frac{\Phi_\text{rob}(x,\hat{x}_s,u^*_{\text{nom}})}{L_gh(x,\hat{x}_s)-e^*_{\nabla h}\|g\|}|\}$, thus,
\begin{align*}
	&e_{\frac{\partial h}{\partial t}}^* +  \alpha( e_h^*)-e_{\nabla h}^*(||f(x)+g(x) u_{\text{nom}}^\ast\| +  |u_\delta|\|g(x)\|)\geq \\ &e_{\frac{\partial h}{\partial t}}^* +  \alpha( e_h^*)-e_{\nabla h}^*(||f(x)+g(x) u_{\text{nom}}^\ast\| +  \Bar{u}_\delta\|g(x)\|).
\end{align*}
This completes the proof.
\end{proof}
Notice that the residual term in Theorem \ref{thm:er_CBF2}, unlike the one in Theorem \ref{thm:er_CBF}, does not depend on the control input, which is the decision variable in the optimization-based control design. In fact, in the residual term of Theorem \ref{thm:er_CBF2}, we have $u_{\text{nom}}^*$ and $\Bar{u}_\delta$ which are not decision variables and hence they are available before solving the optimization problem. Hence, one can utilize a QP to compute the control input instead of a SOCP. Thus, to design an environmentally safe controller with ER-CBFs,  the following QP is considered:
\begin{align}
& u_{\text{rob}}^* = \textrm{arg} \min_{u_{\text{rob}}} \quad  \|u_{\text{rob}} -u_{\text{nom}}^* \|^2 \tag{ER-CBF-QP} \label{ER-CBF-QP}\\ \nonumber
&\textrm{s.t.}   \quad  \frac{\partial h(x,\hat{x}_s)}{\partial t}+L_fh(x,\hat{x}_s)+L_gh(x,\hat{x}_s)u_{\text{rob}} \\&+\alpha(h(x,\hat{x}_s))- e_{\nabla h}^*(\|f(x)+g(x)u_{\text{nom}}^\ast\|+ \bar{u}_\delta\|g(x)\|)\notag\\
&+e_{\frac{\partial h}{\partial t}}^*+ \alpha(e_h^*) \geq 0.  \notag
\end{align}
In the sequel,  the closed-form solution of (\ref{ER-CBF-QP}) tackling \emph{Problem 2} is presented.
\begin{theorem}\label{thm:clform} 
Consider $h(x,x_s)$ as an ER-CBF for system (\ref{eq:system}) with $u \in \mathbb{R}$ and $L_gh(x,\hat{x}_s) \neq 0$. The solution to (\ref{ER-CBF-QP}) is given by $u_{\text{rob}}^\ast = u_{\text{nom}}^\ast + u_{\hat{\delta}}$ where 
\begin{equation}\label{u_delta_hat}
u_{\hat{\delta}} = 
	\begin{cases}
		0&\textrm{if}\quad \widehat{\Phi}_\text{rob}(x,x_s,u_{\text{nom}}^*) \geq 0\\
		-\frac{\widehat{\Phi}_\text{rob}(x,\hat{x}_s,u_{\text{nom}}^*)}{L_gh(x,\hat{x}_s)}  \quad &\textrm{if}\quad \widehat{\Phi}_\text{rob}(x,x_s,u_{\text{nom}}^*) < 0
	\end{cases}
\end{equation}
with $\widehat{\Phi}_\text{rob}(x,\hat{x}_s,u) =  \frac{\partial h(x,\hat{x}_s)}{\partial t}+L_fh(x,\hat{x}_s)+L_gh(x,\hat{x}_s)u+\alpha(h(x,\hat{x}_s))- e_{\nabla h}^*(\|f(x)+g(x)u_{\text{nom}}^\ast\|+ \bar{u}_\delta\|g(x)\|)+e_{\frac{\partial h}{\partial t}}^*+ \alpha(e_h^*)$.
\end{theorem}
\begin{proof}
Notice that in (\ref{ER-CBF-QP}) both objective function and constraint are convex and continuously differentiable with respect to $u_{\text{rob}}$. Hence, one can apply KKT (Karush–Kuhn–Tucker) conditions to provide necessary and sufficient conditions for optimality \cite{boyd2004convex}. The rest of the proof is similar to the proof of Lemma \ref{lem:QP-Closed-form}.
\end{proof}
\begin{remark}
	Note that (\ref{ER-CBF-SOCP}) presents our solution for Problem 1, and (\ref{ER-CBF-QP}) presents our solution for Problem 2. In (\ref{ER-CBF-SOCP}) the optimization constraint ($\Phi_\text{rob}(x,\hat{x}_s,u)$) is imposing safety and environmental robustness simultaneously, so it can be considered as a robust safe filter that modifies the nominal (and possibly unsafe) control input $u_{des}$. On the other hand, in (\ref{ER-CBF-QP}) the optimization constraint ($\widehat{\Phi}_\text{rob}(x,\hat{x}_s,u)$) is only imposing robustness to the nominally safe control input, so it can be regarded as robustness filter of nominally safe input $u^*_{nom}$. In addition to that, computing the control input via (\ref{ER-CBF-QP}) is faster and less prone to infeasiblity than (\ref{ER-CBF-SOCP})\footnote{In MATLAB, the interior point method (which has a polynomial time complexity) is utilized to solve both QP and SOCP (though with different algorithms). However, computing QP i.e. (\ref{ER-CBF-SOCP}) is more time-consuming than computing SOCP, i.e., (\ref{ER-CBF-SOCP}). This is also has been observed in \cite{long2022safe}.}. Furthermore, (\ref{ER-CBF-QP}) (as we shall see in Figures \ref{fig:delx} and \ref{fig:CBF-CLF}) is slightly more conservative than (\ref{ER-CBF-SOCP}). Moreover, it should be noted that, (\ref{ER-CBF-SOCP}) is applicable to the cases where $u\in \mathbb{R}^m$, while (\ref{ER-CBF-QP}) is only used for $u\in \mathbb{R}$.
\end{remark}
\section{Adaptive cruise control example}
\begin{figure}
	\centering
	\includegraphics[scale =.7]{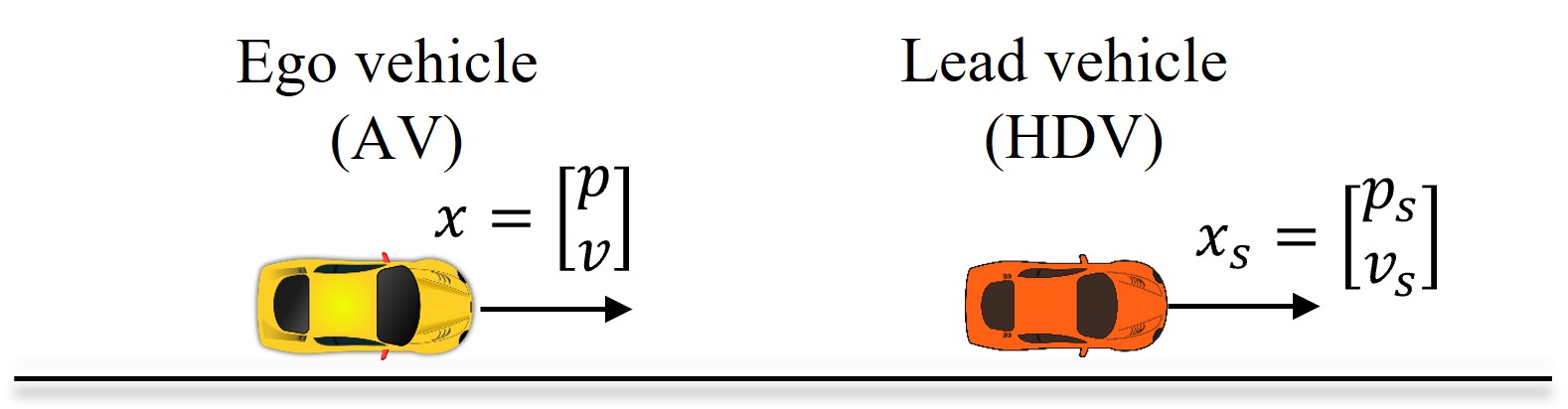}
	\caption{Adaptive cruise control example for mixed traffic where ego vehicle is AV following the lead vehicle which is an HDV.}
	\label{fig:cruise}
\end{figure}
In this section, the proposed robust safe controller is applied to the adaptive cruise control problem where the lead vehicle is a human-driven vehicle (HDV) and the ego vehicle (see Figure \ref{fig:cruise}) which we intend to design control actions for, is an Automated Vehicle (AV).  The following longitudinal dynamics is considered for AVs:
\begin{equation}\label{eq:vehicle_dynamics}
	\begin{aligned}
		\dot{p}&=v,\\
		\dot{v}&=\frac{u}{m}-\frac{F_r(v)}{m},
	\end{aligned}
\end{equation}
where $p$, $v$, and $u$ denote the position, the velocity and the control input for the AV, respectively. The mass of the AV is denoted by $m$ and $F_r(v)$ is the rolling force which is approximated as
\begin{equation}
	F_r(v) = c_0 + c_1v+ c_2v^2,
\end{equation}
with the constants $c_0$, $c_1$, and $c_2$  that can be determined empirically. Equation (\ref{eq:vehicle_dynamics}) can be written in a state-space form as
\begin{equation}\label{eq:vehicle_dynamics_Ctrl_affine}
	\dot{x}
	= \underbrace{\begin{bmatrix}
			v \\
			-\frac{1}{m}F_r \end{bmatrix}}_{f(x)} +\underbrace{
		\begin{bmatrix}
			0 \\
			\frac{1}{m} \end{bmatrix}}_{g(x)}u.
\end{equation}
where $x = [p,v]^\textmd{T}$ represents the state of the AV.  It is assumed that the dynamics of the HDV (the lead vehicle) are unknown and only some uncertain measurements of its state are available.  Let the state of the HDV be given by $p_s$ and $v_s$, denoting the HDV's position and velocity, respectively. To ensure the rear-end safety between AV and HDV, the following constraint should be imposed:
\begin{equation}
	p_s - p \geq  T_hv,
\end{equation}
where $T_h$ is a look ahead (reaction) time. To enforce this constraint in the synthesis of the control input of the AV, the following CBF is considered:
\begin{equation}\label{eq:cbf_veh}
	h(x,x_s) = p_s-p -T_hv -\frac{1}{2}\frac{(v_s-v)^2}{c_dg},
\end{equation}
with $x_s = [p_s,v_s]^\textmd{T}$, and ${c_dg}$ is the maximum deceleration with $g$ and $c_d$  being the gravity acceleration and deceleration factor, respectively. It is assumed that, instead of the exact value of $x_s$,  an estimate of it, $\hat{x}_s$, is available to the AV through its onboard sensors or a road-side coordinator. Therefore, it follows that $p_s =\hat{p}_s + e_p$, $v_s =\hat{v}_s + e_v$ and $\dot{v}_s =\dot{\hat{v}}_s + e_{\dot{v}}$. Then, by assuming known uncertainty bounds between $x_h$ and $\hat{x}_s$, it follows that
\begin{equation}
	|p_s - \hat{p}_s| \leq \mathcal{E}_p, \:  |v_s - \hat{v}_s| \leq \mathcal{E}_v,\:  |\dot{v}_s - \dot{\hat{v}}_s| \leq \mathcal{E}_{\dot{v}}.
\end{equation}
Knowing the worst-case bounds on the errors of state of HDV, we need to quantify the  error in the designed CBF (\ref{eq:cbf_veh}), its gradient,  and its time-derivative as follows:
\begin{align*}
	\qquad e_h(x,x_s,\hat{x}_s) &=  h(x,x_s) -  h(x,\hat{x}_s) \\
	&= e_p +\frac{2e_v(\hat{v}_s-v) +e_v^2}{2c_dg},\\
	\qquad	\quad e_{\nabla h}(x,x_s,\hat{x}_s) &=  \|\nabla h(x,x_s) - \nabla h(x,\hat{x}_s) \|\\ &=\|[0,\frac{e_v}{c_dg}]^T\| = | \frac{e_v}{c_dg}|,\\
	\qquad e_{\frac{\partial h}{\partial t}}(x,x_s,\hat{x}_s) &= \frac{\partial h(x,x_s)}{\partial t} - \frac{\partial h(x,\hat{x}_s)}{\partial t}\\
	&=v_s-\frac{\dot{v}_s(v_s-v)}{c_dg}- \hat{v}_s+\frac{\hat{\dot{v}}_s(\hat{v}_s-v)}{c_dg} \\
	&=e_v-\frac{\dot{v}_se_v+e_{\dot{v}}(\hat{v}_s-v) +e_ve_{\dot{v}}}{c_dg}.
\end{align*}
Now $e_{\frac{\partial h}{\partial t}}^*$, $e_h^*$, and $e_{\nabla h}^*$ can be obtained via (\ref{eq:e_h}). They will be utilized to synthesize the environmentally safe control input for the adaptive cruise control. To compute the desired control input to achieve the control objective,  the following CLF is defined:
\begin{equation*}
	V(x) = (v-v_d)^2,
\end{equation*}
where $v_d$ is the desired velocity on the road. To apply maximum and minimum permissible velocity on the road ($v_{\max}$ and $v_{\min}$),  the following CBFs are defined.
\begin{align*}
h_2(x) &= v_{\max} - v,\\
h_3(x) &= v - v_{\min}.
\end{align*}
 Note that these CBFs only depend on the state of the AV, thus, they will not be used as ER-CBF constraints while implementing them in a QP or SOCP and  are considered as a regular CBF constraint. 
 In the next section,  the control input is designed for this example using (\ref{CBF-QP}), (\ref{ER-CBF-SOCP}) and (\ref{ER-CBF-QP}) and the obtained results are compared.
\section{Simulation Results}

To demonstrate the efficacy of the proposed methods, several simulations are carried out  in MATLAB on the adaptive cruise control example presented in the previous section. MATLAB QUADPROG  and CONEPROG commands are used for solving QPs and SOCPs and ODE45 is used for integrating AV dynamics. To generate the HDV's motion, the so-called linear free-flow model is used \cite{ahmed1999modeling}:
\begin{equation}
	\dot{v}_s(t) = \lambda[v^{(d)} (t) -v_s(t-\tau)] + \epsilon(t),
	\label{eq:free_flow}
\end{equation} 
where $v^{(d)}$ is the desired road velocity, $\tau$ is the HDV driver's reaction time, $\lambda$ is a constant, and $\epsilon$ is a zero-mean Gaussian noise with the variance of $\sigma$.
For the simulations the parameters presented in Table \ref{tab:veh_params} are for the vehicle dynamics for AV and HDV, and $g = 9.81 \, \textrm{m/s}^2$, $ c_d = 0.3$.  To solve the QP and SOCP,  CLF and CBF rates are selected  equal as $\gamma = \nu=  5$. In addition,  the maximum and minimum admissible road speed are considered as $v_{\max} = 120 \textrm{km/h}$ and $v_{\min} = 60 \textrm{km/h}$.
\begin{table}[h!]
	\centering
	\begin{tabular}{||c|c||c|c||}
		\hline
		\multicolumn{2}{||c||}{AV}  & \multicolumn{2}{c||}{HDV }\\
		\hline
		Parameter  & Value &  Parameter & Value \\
		\hline 
		$m$ & 1650 kg & $\lambda$ & 0.309\\
		$c_0$ & 0.1 N & $\sigma$ & 1.13\\
		$c_1$ & 5 $\textrm{Ns/m}$ & $v^{(d)}$ & 100 km/h\\
		$c_2$ & 0.25  $\textrm{Ns}^2 \textrm{/m}$ & $\tau$ &0\\
		\hline
	\end{tabular}
	\caption{Vehicles parameters used in the simulations.}
	\label{tab:veh_params}
\end{table}
To begin the simulation, it is assumed that initial distance between the ego and lead vehicles is $\Delta p = 80 \textrm{m}$ and $v= v_s = 27.8 \, \textrm{m/s}= 100 \, \textrm{km/s}$. Moreover, the worst-case errors for HDV state are considered as $\mathcal{E}_p = 1 \textrm{m}$, $\mathcal{E}_v = 1 \, \textrm{m/s}$ and $\mathcal{E}_{\dot{v}} = 0$. 
In Figure \ref{fig:uncertainty_bounds}, the  uncertainty bound for CBF defined in (\ref{eq:cbf_veh}) is shown using (\ref{CBF-QP}), (\ref{ER-CBF-SOCP}) and (\ref{ER-CBF-QP}). It can be observed from the figure that in the case of \ref{CBF-QP} (the top plot with red color), even though the nominal CBF remains non-negative, the uncertainty bound crosses the $x$-axis and the CBF becomes negative; this is clearly not a safe behavior and the forward invariance property is no longer guaranteed. On the other hand, CBFs obtained via (\ref{ER-CBF-SOCP}) and (\ref{ER-CBF-QP}) remain non-negative in the presence of the uncertainties. It is worth noting that the nominal CBF for (\ref{ER-CBF-SOCP}) and (\ref{ER-CBF-QP}) is more conservative\footnote{ When a CBF has farther distance from x-axis (zero) it is considered to be more conservative.} than (\ref{CBF-QP}) which is indeed expected.
\begin{figure}
	\centering
	\includegraphics[scale =.6]{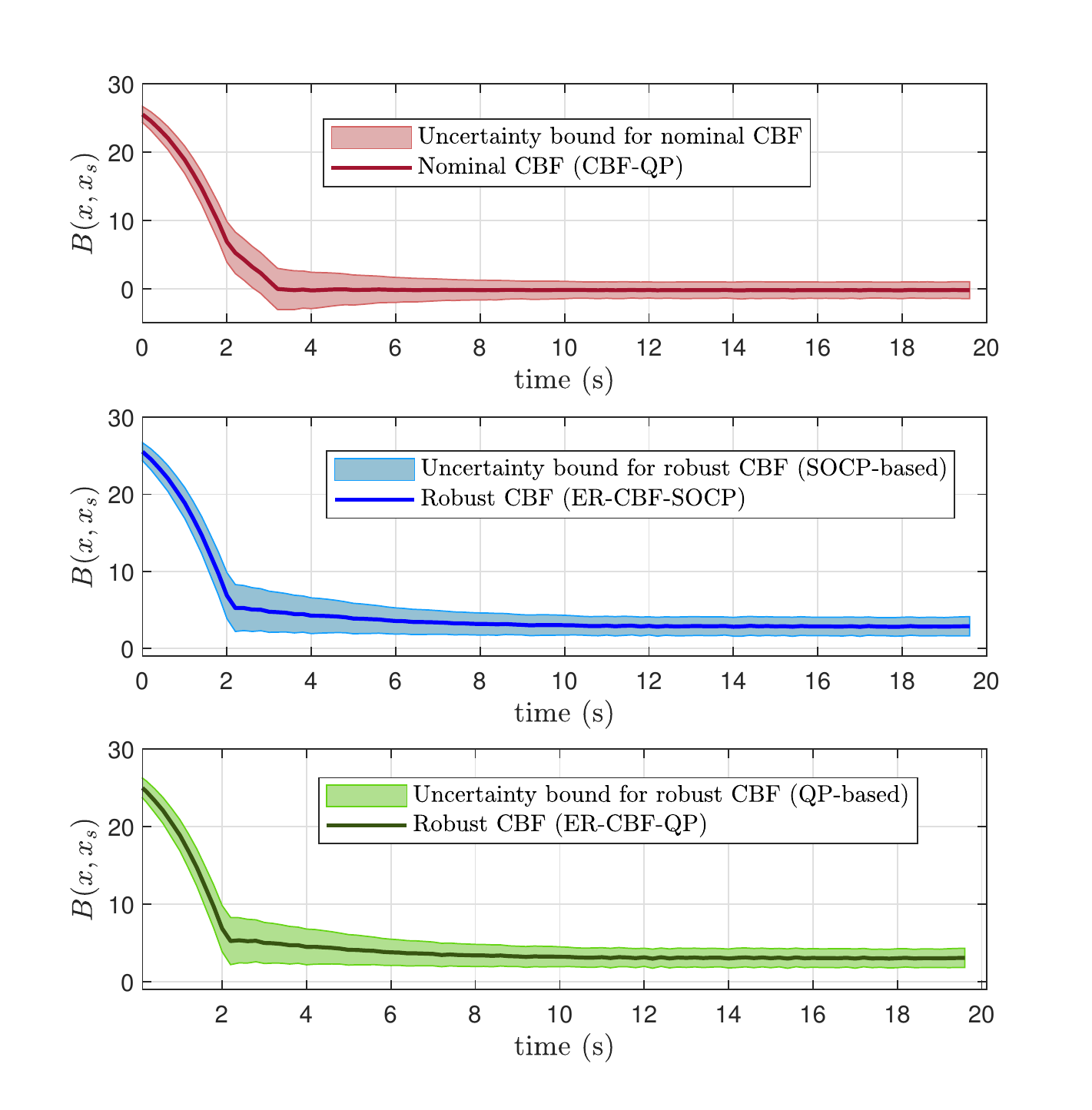}
	\caption{CBF and its uncertainty bound for keeping safe distance in cruise control using nominal CBF (top), robust CBF with SOCP (middle), and robust CBF with QP (bottom).}
	\label{fig:uncertainty_bounds}
\end{figure}

In Figure \ref{fig:delx}, the velocity of AV and HDV, their distance ($\Delta p = p_s-p$) and the control input for AV are shown using three different methods. As it can be seen from the figure, the results of (\ref{ER-CBF-SOCP}) and (\ref{ER-CBF-QP}) almost coincide, and as it is shown in the magnified cross-section of distance plot, (\ref{ER-CBF-QP}) is slightly more conservative. 
\begin{figure}
	\centering
	\includegraphics[scale =.6]{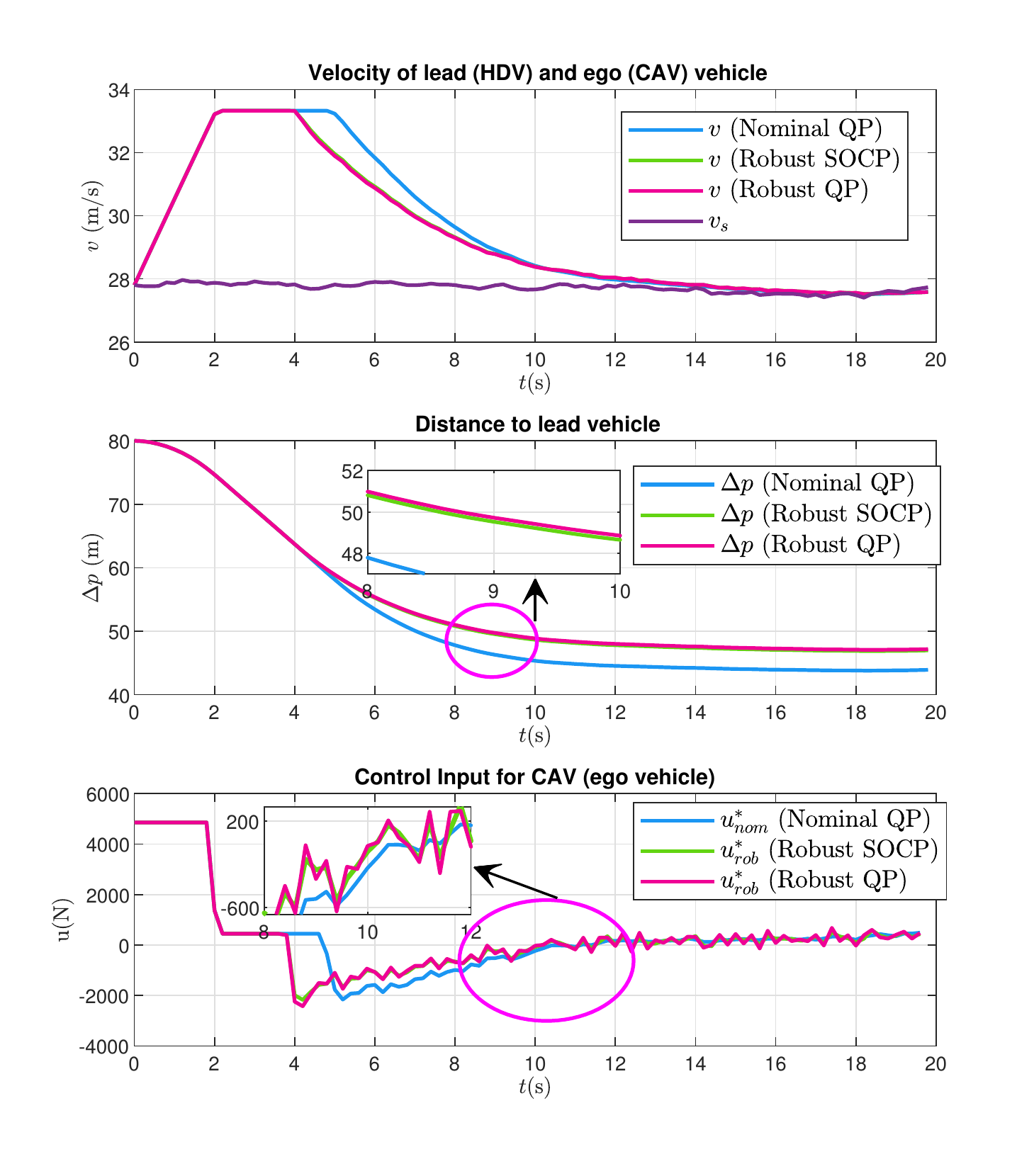}
	\caption{Velocity and distance of ego and lead vehicles in the adaptive cruise control along with the control input for ego vehicle using nominal CBF, robust CBF with SOCP, and robust CBF with QP.}
	\label{fig:delx}
\end{figure}

Figure \ref{fig:CBF-CLF} also shows the CBF and CLF plot for the case of Figure \ref{fig:delx}. Similar to Figure \ref{fig:delx}, the results of (\ref{ER-CBF-SOCP}) and (\ref{ER-CBF-QP}) are coinciding, and the results of (\ref{ER-CBF-QP}) are slightly more conservative. 
\begin{figure}
	\centering
	\includegraphics[scale =.6]{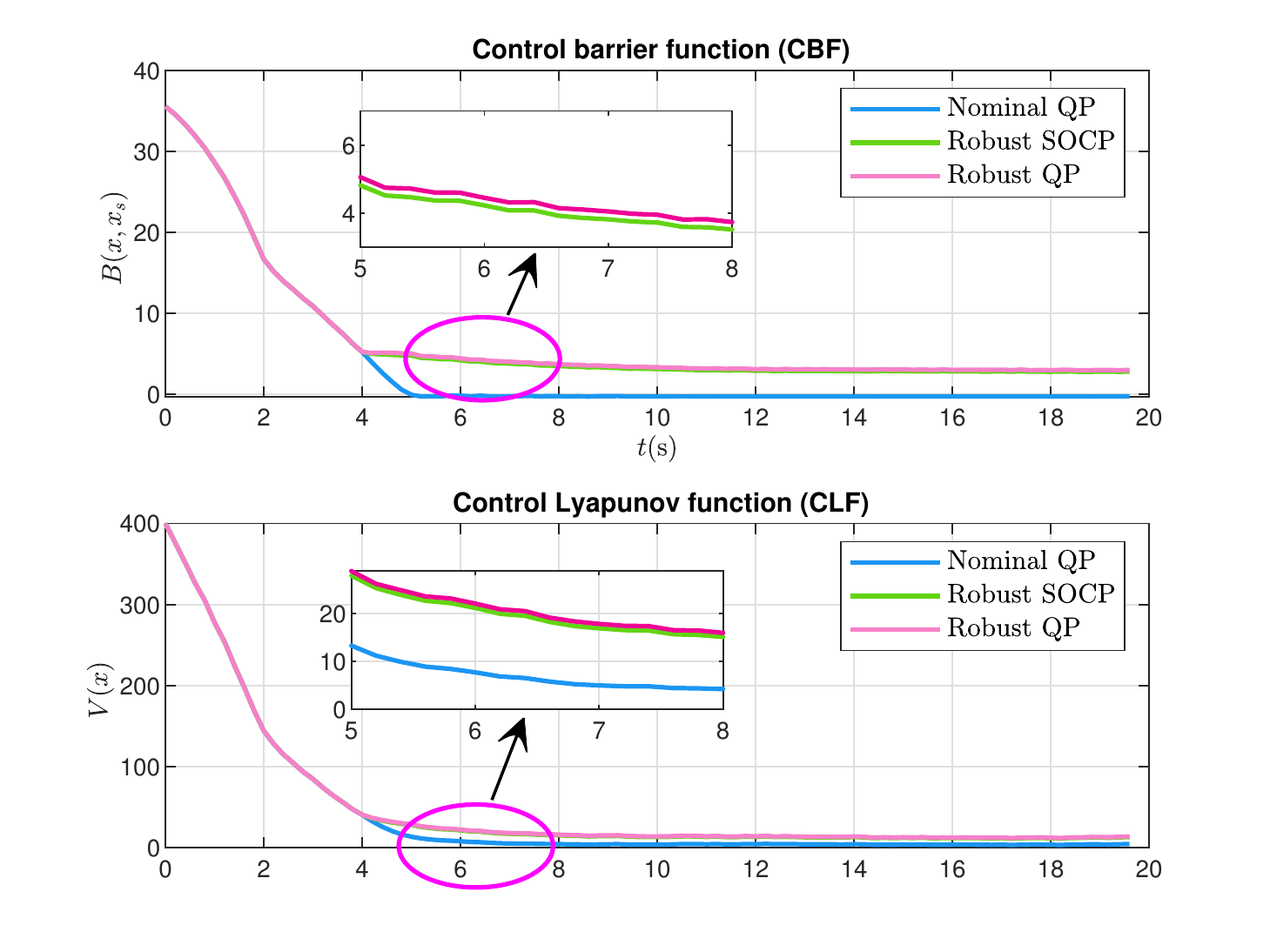}
	\caption{CBF for keeping safe distance with lead vehicle and CLF for achieving desired distance with lead vehicle in cruise control using nominal CBF, robust CBF with SOCP, and robust CBF with QP.}
	\label{fig:CBF-CLF}
\end{figure}
In Figure \ref{fig:ER-CBF-QP} the plot for the control input resulting from the closed form solution (\ref{ER-CBF-QP}) using Theorem \ref{thm:clform} is presented. As it can be seen from the figure, at each time instant, $u_{\text{nom}}^*$ and $u_{\hat{\delta}}$ are computed and then $u_{\text{rob}}^*$ is computed via $u_{\text{rob}}^* = u_{\text{nom}}^* + u_{\hat{\delta}} $ where $u_{\hat{\delta}}$ is defined in \eqref{u_delta_hat} in Theorem \ref{thm:clform}. 
\begin{figure}[thpb]
	\centering
	\includegraphics[scale =0.6]{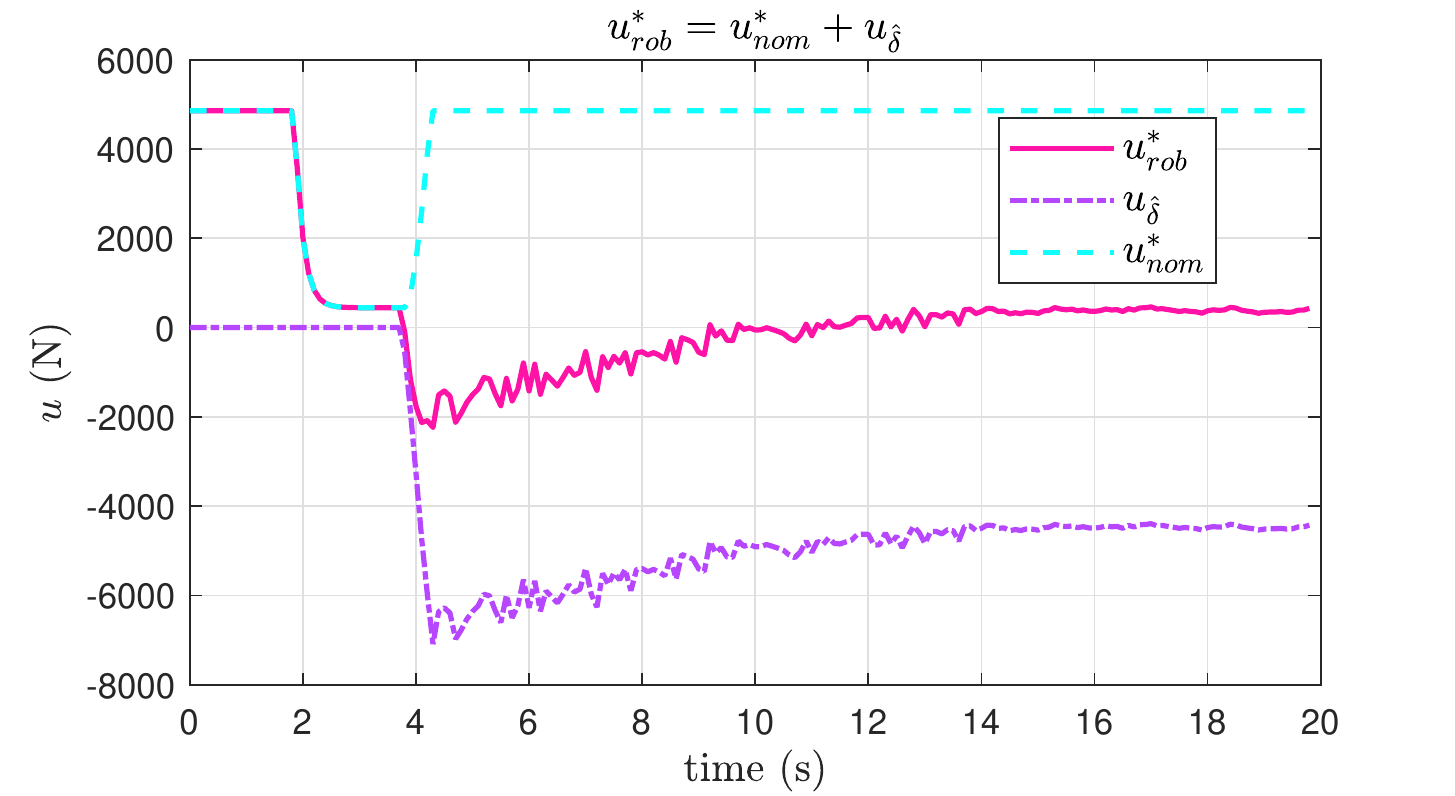}
	\caption{Obtaining robust control input via closed form solution of robust QP.}
	\label{fig:ER-CBF-QP}
\end{figure}
\section{conclusion}
The control synthesis problem using environmentally robust control barrier functions is considered in this paper and it is shown that accounting for the worst-case error in a dynamical environment results in the ER-CBF-SOCP formulation. Then, we present the ER-CBF-QP alternative which minimally modifies the nominal safe input resulting from the solution of CBF-QP. ER-CBF-SOCP and ER-CBF-QP results are almost similar. However, ER-CBF-QP has a better computational time and is less prone to infeasibility. A future direction would be extending ER-CBF-QP for multi-input dynamical systems.

%
%




%
%


 \bibliographystyle{ieeetran} 
\bibliography{JournalPaper.bib}

\end{document}